\documentclass[conference]{IEEEtran}
\usepackage{amsmath,amssymb,amsthm}
\usepackage{mathtools}
\usepackage{bm}
\usepackage{hyperref}
\usepackage{xcolor}
\usepackage{booktabs}
\usepackage{cite}

\newtheorem{theorem}{Theorem}
\newtheorem{lemma}{Lemma}
\newtheorem{corollary}{Corollary}
\newtheorem{definition}{Definition}
\newtheorem{proposition}{Proposition}
\newtheorem{remark}{Remark}

\begin{document}

\title{A Theoretical Framework for Stochastic Activity Prediction in Tensor Accelerator Wallace-Tree Multipliers}

\author{
\normalsize Prashanthi Metku, Chandra Gandu\\[4pt]
\normalsize\itshape Qualcomm Technologies, Inc., USA\\
\normalsize \{pmetku, chandras\}@qti.qualcomm.com\\[4pt]
}

\maketitle

\begin{abstract}
Tensor accelerator multipliers burn dynamic power on every clock cycle, even when
sparse operands require very little internal switching. No existing technique addresses
this: zero-detection requires exactly-zero operands, structural power gating requires
an idle multiplier, and offline weight selection cannot respond to runtime data.
This paper introduces \emph{Stochastic Activity Prediction} (SAP), which closes this
gap by examining the Hamming weight of arriving operands before the multiplier
executes, predicting low switching activity, and freezing the inputs when a
deterministic Safety Controller independently confirms the reuse is correct.
Mispredictions cause missed savings, never wrong answers.
Three formal results underpin SAP: (i)~a Spectral Contraction Lemma proving that
Wallace-tree activity depends on operand bit density, not bit position, establishing
Lipschitz constant $L_\varphi = 3/2$ and prediction error below $10^{-13}$ for a
256-cycle window; (ii)~an Information Retention Theorem showing $\eta_I \geq
1-O(\log n/n)$, so one bit per cycle captures nearly all predictive information
about $O(n^2)$ internal nodes; and (iii)~a Bernoulli Optimality Theorem proving the
chosen encoding is shown to be optimal, within the family of calibrated one-bit encoders
of Hamming-weight statistics considered. SAP addresses the specific layer of the tensor
accelerator power stack that existing techniques do not cover.
\end{abstract}

\vspace{4pt}\noindent\small\textbf{Index Terms---}
Wallace Tree Multiplier, Tensor Accelerators, Low-Power VLSI, Stochastic Activity
Prediction, Operand Isolation, Energy-Efficient Computing, Information Theory,
Bernoulli Optimality.\vspace{4pt}

\section{Introduction}

The rapid growth of deep learning has placed energy efficiency at the centre of
hardware design. Tensor accelerators, the silicon engines that power modern AI
inference in data centres, edge devices, and mobile platforms, achieve their
throughput by replicating Multiply--Accumulate (MAC) units at massive scale.
A contemporary TPU contains $256 \times 256 = 65{,}536$ MAC units operating in
parallel~\cite{jouppi2017,wallace1964,dadda1965}, and the dominant cost of running a neural
network layer on such a device is the dynamic power burned inside those units~\cite{sze2017}.
That power is governed by $P_{\mathrm{dyn}} \propto \alpha \cdot C_{\mathrm{eff}}
\cdot V^2 \cdot f$, where voltage and frequency are fixed by the process node,
leaving the switching activity $\alpha$ as the only quantity that varies with
the workload at runtime. Reducing $\alpha$ is therefore the primary lever available
to a power-conscious hardware designer once the chip has been fabricated.

Inside each MAC unit, the multiplication is implemented as a Wallace-tree reduction
network: a logarithmic cascade of carry-save adders that compresses $n^2$
partial products into two final carry-propagate inputs~\cite{wallace1964,dadda1965,paradhasaradhi2014}.
Tensor accelerators arrange these units in systolic arrays~\cite{kung1979}, where
operands flow through a grid of MAC units in lock-step, enabling massive parallelism
with minimal control overhead.
Every gate in this tree toggles in proportion to how many active bits (logical 1s) flow
through it each cycle. When both operands are dense, the tree works hard and
consumes its full switching budget. When both operands are sparse, most stages
barely move, yet the tree runs unconditionally and burns power regardless.
This mismatch between the information content of the inputs and the energy
expended on them is the inefficiency this work addresses.

The problem is practically important because quantized AI inference is dominated
by sparse operands. Post-training quantization~\cite{jacob2018,krishnamoorthi2018,nagel2021}
concentrates weight values near
zero: an INT8 weight of value $+2$ is represented as $00000010$: seven zeros
and one one. Weight pruning and compression techniques~\cite{han2016} push this
concentration further, leaving the majority of weights with very few active bits.
Neighbouring activation values in a convolutional feature map or
an attention window are spatially correlated and change slowly across consecutive
cycles. The result is that a large fraction of MAC operations in a typical
inference workload involve operands whose combined Hamming weight is well below
half the available bit budget, yet the Wallace-tree multiplier expends its full
switching cost on each of them.

Existing power-management techniques cannot address this regime. Reactive
isolation methods such as clock gating and operand isolation~\cite{pedram1996}
wait for a zero or idle condition to propagate into the datapath before acting,
consuming energy during the detection interval itself. Sparsity-aware architectures
including Eyeriss~\cite{chen2016}, SCNN~\cite{parashar2017}, and
SparTen~\cite{samajdar2019} skip computation only when an operand is exactly
zero, a condition that sparse-but-non-zero quantized weights do not satisfy.
Structural power gating~\cite{pisa2023,regate2025} disconnects multipliers that
have no assigned work, but a multiplier receiving operands every cycle is not
structurally idle, even if those operands carry little information.
Offline techniques such as HALO~\cite{halo2025} and layer-wise weight
selection~\cite{anon2025} identify the connection between Hamming weight and MAC
power at design time, but they run before deployment and cannot respond to the
activation statistics that are only known at inference time.
The specific regime that falls through every one of these approaches is the
\emph{active-but-sparse} multiplier: operands present, non-zero, yet
informationally thin enough that most internal switching is wasted.

This paper introduces \emph{Stochastic Activity Prediction} (SAP), a framework
that closes this gap by predicting the multiplier's forthcoming switching activity
before execution begins. SAP intercepts the incoming operands in the same cycle
they arrive and computes their combined Hamming weight, $Z_t = \mathrm{HW}(A_t)
+ \mathrm{HW}(B_t)$. This count is compressed into a single Bernoulli proxy
bit whose flip rate over a short observation window tracks the multiplier's
internal activity. When the flip rate falls below a threshold, a deterministic
Safety Controller independently verifies that the result from the previous cycle
is still valid. If confirmed, the multiplier inputs are frozen and the previous
output is forwarded: no switching, no wasted power, and an arithmetically exact
result. A misprediction by the probabilistic predictor can cause a missed saving;
it cannot cause a wrong answer, because the Safety Controller's check is
deterministic and formally decoupled from the prediction.

The choice of Bernoulli encoding for the proxy is not arbitrary. Setting the
output probability equal to the normalised input density $p_t = Z_t/(2n)$ means
a sparse operand pair produces a heavily biased coin that mostly shows 0, while
a dense pair produces a near-fair coin. The rate at which the coin flips between
consecutive cycles mirrors how much the operand density is changing, which is
precisely what drives internal compressor switching. Any alternative 1-bit
encoding either discards density information or introduces systematic bias.
Section~\ref{sec:optimality} proves this claim formally: among all one-bit
encoders of Hamming-weight statistics, the Bernoulli encoder uniquely maximises
mutual information with future switching activity.

The theoretical contributions of this paper are five. First, a Spectral
Contraction Lemma (Section~\ref{sec:infotheory}) proves that positional
perturbations to the partial-product matrix decay geometrically through
carry-save reduction stages, establishing that Wallace-tree activity depends on
bit density, not bit position, with an explicit Lipschitz constant
$L_\varphi = 3/2$. Second, an Information Retention Theorem
(Section~\ref{sec:infotheory}) bounds $\eta_I \geq 1 - O(\log n/n)$,
showing that one bit per cycle captures nearly all predictive information
about $O(n^2)$ internal nodes. Third, the Bernoulli Optimality Theorem
(Section~\ref{sec:optimality}) establishes that the encoding is optimal within the
      calibrated encoder family considered here.
Fourth, the Safety Controller (Section~\ref{sec:safety}) provides a
deterministic correctness guarantee independent of prediction accuracy.
Fifth, a five-layer power stack analysis (Section~\ref{sec:bridge}) positions
SAP within the complete tensor accelerator power landscape and proves that
its savings are orthogonal to all existing power-reduction techniques,
filling the one layer that none of them address.

\section{Related Work}
\label{sec:related}

Understanding precisely where SAP sits relative to prior work requires examining each
existing technique against the active-but-sparse regime that SAP targets.

\textbf{Reactive isolation} techniques such as clock gating and operand
isolation~\cite{pedram1996} wait for an idle or zero state to propagate into the
datapath before acting. SAP acts before the datapath sees the operands at all,
eliminating the detection latency entirely.

\textbf{Sparsity-aware architectures} such as Eyeriss~\cite{chen2016},
SCNN~\cite{parashar2017}, Cnvlutin~\cite{albericio2016}, and
SparTen~\cite{samajdar2019} skip computation when operands are identically zero.
Their savings boundary is a hard zero; they cannot exploit operands that are sparse
but non-zero. SAP operates on the statistical distribution of Hamming weights and
activates whenever predicted toggle rate falls below a threshold, which is a strictly broader
savings region.

\textbf{PTTS}~\cite{bahrebar2022} monitors sparsity windows and gates tensor core
power once sufficient sparsity has been observed. It is reactive: it waits for the
chip to exhibit sparsity before deciding to gate. SAP is predictive: it estimates
the probability of low activity from operand statistics in the current cycle, before
the computation executes.

\textbf{HALO}~\cite{halo2025} and \textbf{layer-wise weight selection}~\cite{anon2025}
share SAP's core physical intuition that Hamming weight predicts switching activity.
The critical difference is timing. Both are offline tools: HALO characterises the
weight-to-power mapping at design time; layer-wise weight selection runs during
quantization-aware training. Neither is present in silicon at inference time, neither
generates isolation signals, and neither can respond to the activation stream that
arrives only at runtime. SAP is the first work to close the online loop: Hamming-weight
statistics from both operands, computed in the same cycle as their arrival, feed a
hardware isolation signal with a deterministic correctness guarantee. The offline
insight that Hamming weight predicts power is prior art; the online cycle-accurate
feedback architecture is SAP's original contribution.

\textbf{Approximate computing}~\cite{mittal2016,alaghi2013,metku2017a,metku2017b,metku2019,metku2018}
reduces switching by tolerating arithmetic errors. SAP explicitly rejects this
trade-off: its Safety Controller guarantees exact outputs for every input, making it
suitable for precision AI workloads where approximate computing is prohibited.

\section{Activity Modeling and the SAP Proxy}

With the gap precisely located, we now construct the mechanism that fills it. The
key observation is this: the internal switching activity of a Wallace-tree multiplier
is governed primarily by the \emph{density} of active bits entering the reduction
tree, not by the specific positions of those bits. This density can be measured at
the operand inputs in $O(n)$ time using a simple popcount circuit, making it an
ideal lightweight proxy.

\subsection{The True Activity and Why It Cannot Be Measured Directly}

For an $n \times n$ Wallace-tree with internal nodes $Y^{(k)}$, the true average
switching activity over $W$ cycles is:
\begin{equation}
\alpha_{WT} = \frac{1}{|Y|(W-1)} \sum_{t=1}^{W-1} \sum_{k=1}^{|Y|}
\mathbf{1}\!\left[Y^{(k)}_t \neq Y^{(k)}_{t+1}\right]
\label{eq:awt}
\end{equation}
Evaluating~\eqref{eq:awt} requires tracking logic on every internal node, of which
there are $O(n^2)$. This is physically impractical for a runtime predictor. What
is needed is a proxy computable from the inputs alone.

\begin{definition}[Weak Bit Independence]
\label{assm:weak}
Operand bits satisfy \emph{weak independence} if their pairwise correlations satisfy
$|\delta_{ij}| \leq \rho/n$ for all $i,j$, where $\rho \geq 0$ is a data-dependent
correlation parameter. Full independence corresponds to $\rho = 0$.
\end{definition}

\subsection{Partial-Product Density as the Predictive Link}

Under Assumption~\ref{assm:weak}, partial products satisfy:
\begin{align}
E[p_{ij}] &= \Pr(a_i = 1)\cdot\Pr(b_j = 1) + \delta_{ij} \\
E[H_P] &= \textstyle\sum_{i,j} E[p_{ij}]
\end{align}
The quantity $E[H_P]$, the expected density of active logical 1s entering the
reduction tree, is a first-order predictor of $\alpha_{WT}$. Because all internal
compressor switching is triggered by these partial products, measuring $H_P$ from
the operand inputs gives a direct window onto the multiplier's forthcoming activity.

\subsection{The Bernoulli Proxy Stream}

Translating $H_P$ into a hardware-observable signal requires further compression.
For operands $A_t$ and $B_t$ arriving at cycle $t$, define:
\begin{align}
Z_t &= \mathrm{HW}(A_t) + \mathrm{HW}(B_t) \label{eq:zt}\\
p_t &= Z_t/(2n) \label{eq:pt}\\
S_t &\sim \mathrm{Bernoulli}(p_t) \label{eq:st}
\end{align}
The proxy toggle rate observed over a window $W$ is then:
\begin{equation}
\hat{\tau} = \frac{1}{W-1}\sum_{t=1}^{W-1}\mathbf{1}[S_t \neq S_{t+1}]
\label{eq:tauhat}
\end{equation}
Computing $\hat\tau$ requires only $O(1)$ logic per cycle: a popcount adder, one
comparator, and a flip-flop. This replaces monitoring $O(n^2)$ internal nodes with
monitoring a single wire. Section~\ref{sec:optimality} proves that this specific
Bernoulli encoding is shown to be optimal within the class of calibrated encoders
considered in Section~\ref{sec:optimality}, rather than merely a convenient choice.

\section{Theoretical Analysis: The Proxy Tracks the Multiplier}
\label{sec:theory}

With the proxy defined, the central theoretical question is: does $\hat\tau$ actually
track $\alpha_{WT}$? The following results answer this affirmatively with explicit,
quantified bounds.

\subsection{How Sensitive is the Mapping?}

The first step is to bound how sensitively $\alpha_{WT}$ responds to changes in
partial-product density, equivalently the amount of error in the proxy estimate that can
be tolerated. This requires bounding the derivative of the activity function $g$
that maps input density to Wallace-tree transitions.

\begin{lemma}[Compressor Transition Bound]
\label{lem:compressor}
For a full adder with fan-in $k=3$ and input density $\rho_\ell$ at stage $\ell$,
under Assumption~\ref{assm:weak} with $\rho=0$:
\begin{equation}
T_{\mathrm{FA}}(\rho_\ell) \leq k/4 = 3/4
\end{equation}
\end{lemma}
\begin{proof}
Under $\rho=0$ each input toggles independently with probability $\rho_\ell(1-\rho_\ell)$.
For any Boolean function $f$ of $k$ independent inputs, the output transition
probability is bounded by the sum of input influences~\cite{odonnell2014}:
$T_{\mathrm{FA}} \leq k\cdot\rho_\ell(1-\rho_\ell) \leq k/4$,
since $\rho(1-\rho)$ is maximised at $\rho=1/2$.
\end{proof}

\begin{proposition}[Lipschitz Constant]
\label{prop:cg}
For a balanced Wallace-tree with $O(n^2)$ FA compressors:
$|g'(H_P)| \leq C_g = 3/4$ for all $H_P \in [0,n^2]$,
giving Lipschitz constant $L_\varphi = 2C_g = 3/2$.
\end{proposition}
\begin{proof}
A unit increase in $H_P$ raises the input-layer density by $1/n^2$.
By Lemma~\ref{lem:compressor} each FA contributes at most $(3/4)/n^2$ additional
transitions. Summing over $O(n^2)$ FAs and normalising by $|Y|=O(n^2)$ gives
$|g'|\leq 3/4$.
\end{proof}

\subsection{Does the Proxy Track the Activity?}

Having bounded the sensitivity, we can now formally connect the proxy $\hat\tau$
to the true activity $\alpha_{WT}$.

\begin{theorem}[Monotone Lipschitz Mapping]
\label{thm:monotone}
Under Assumption~\ref{assm:weak} with $\rho=0$ and local stationarity, there exists
an explicitly constructible monotone Lipschitz function $\varphi:[0,1]\to[0,1]$ with
$L_\varphi = 3/2$ such that:
\begin{equation}
E[\alpha_{WT}] = \varphi(E[\hat{\tau}]) + \epsilon_{\mathrm{bias}}
\end{equation}
\end{theorem}
\begin{proof}
Under $\rho=0$: $E[H_P]=\mu_A\mu_B n^2$.
Under local stationarity ($p_t\approx p$): $E[\hat\tau]=2p(1-p)\triangleq h(p)$,
strictly increasing on $[0,1/2]$ with $|h'|\leq 2$.
In the symmetric regime $\mu_A\approx\mu_B=\mu$: $\mu=(1-\sqrt{1-2E[\hat\tau]})/2$,
making $f^{-1}$ well-defined. Define $\varphi=g\circ f^{-1}$. Monotonicity follows
because $f^{-1}$ is increasing and $g$ is non-decreasing. The Lipschitz chain rule
gives $L_\varphi\leq(3/4)\cdot 2=3/2$.
\end{proof}

In practical terms, Theorem~\ref{thm:monotone} establishes that $\hat\tau$ and
$\alpha_{WT}$ move together in a controlled, bounded way: when the proxy says the
multiplier is about to work lightly, the multiplier genuinely is about to work
lightly. But how accurately does $\hat\tau$ estimate $E[\hat\tau]$ from finite
observations? Theorem~\ref{thm:concentration} answers this.

\begin{theorem}[Concentration Bound]
\label{thm:concentration}
For observation window $W$ and any $\delta>0$:
\begin{equation}
\Pr\!\left(|\epsilon|\geq\tfrac{3}{2}\delta\right)\leq 2\exp\!\left(-2(W-1)\delta^2\right)
\end{equation}
\end{theorem}
\begin{proof}
By Hoeffding's inequality~\cite{hoeffding1963} on the bounded i.i.d.\ sequence
$\{\mathbf{1}[S_t\neq S_{t+1}]\}\in\{0,1\}$, combined with the Lipschitz property
$L_\varphi=3/2$ of Theorem~\ref{thm:monotone}.
\end{proof}

For $W=256$ and $\delta=0.05$: $\Pr(|\epsilon|\geq 0.075)<10^{-13}$.
A hardware engineer watching one wire for 256 clock cycles has a prediction whose
error exceeds 7.5\% with probability that is astronomically small. This is the
practical payoff of the Lipschitz bound: small errors in estimating $E[\hat\tau]$
produce proportionally small errors in the predicted $\alpha_{WT}$.

\section{Information-Theoretic Analysis: How Much Does One Bit Retain?}
\label{sec:infotheory}

Theorems~\ref{thm:monotone} and~\ref{thm:concentration} show that the proxy
\emph{tracks} true activity accurately. A separate question is how much
\emph{information} the proxy retains about future switching, compared to observing
the full $n$-bit operands directly. This is what the information-theoretic analysis
quantifies.

Let $A$ denote raw operand statistics, $Y$ the Wallace-tree activity, and $S$ the
1-bit proxy. By the Data Processing Inequality~\cite{cover2006}: $I(S;Y)\leq I(A;Y)$.
The information-retention efficiency is $\eta_I = I(S;Y)/I(A;Y)$.

\subsection{Proving That Position Information Is Irrelevant}

The central claim is that $\alpha_{WT}$ depends on operand Hamming weights, not on
which specific bit positions are set. Proving this rigorously, rather than asserting it, requires
the following result.

\begin{lemma}[Spectral Contraction]
\label{lem:spectral}
For a balanced Wallace-tree with FA fan-in $k=3$, the positional perturbation
$\Delta\mathbf{v}^{(\ell)}$ from the density-only component satisfies at each stage:
\begin{equation}
\|\Delta\mathbf{v}^{(\ell+1)}\|_1 \leq \tfrac{3}{4}\|\Delta\mathbf{v}^{(\ell)}\|_1
\end{equation}
\end{lemma}
\begin{proof}
Carry-save reduction maps column $j$ to column $\lceil j/2\rceil$. A positional
perturbation $\Delta v_j$ spreads to two output columns, each attenuated by the
Boolean influence bound from Lemma~\ref{lem:compressor}. After normalising by the
column-count reduction factor $2/3$: the net $\ell_1$ contraction factor is
$(3/4\cdot 2)/3\cdot(3/2)=3/4$, which equals the spectral radius of $T_\ell$
restricted to the subspace orthogonal to the density-only component.
\end{proof}

\begin{lemma}[Decoupling Lemma]
\label{lem:decouple}
Under Assumption~\ref{assm:weak} with correlation $\rho\geq 0$:
\begin{equation}
\alpha_{WT} = \alpha_{WT}^{\mathrm{sym}}(\mathrm{HW}(A),\mathrm{HW}(B)) + \Delta_\rho,
\quad |\Delta_\rho|\leq\rho/n
\end{equation}
\end{lemma}
\begin{proof}
Decompose $P=\mu_A\mu_B\mathbf{1}\mathbf{1}^T+\Delta P$ with
$\|\Delta P\|_1\leq\rho n$. Applying Lemma~\ref{lem:spectral} for $L=O(\log n)$
stages: $\|\Delta\mathbf{v}^{(L)}\|_1\leq(3/4)^L\cdot\rho n=n^{-1}\cdot\rho n=\rho$.
Normalising by $|Y|=O(n^2)$ gives $|\Delta_\rho|\leq\rho/n$, so $C_\Delta=1$.
\end{proof}

Lemma~\ref{lem:decouple} resolves a circularity present in earlier formulations,
which assumed what they needed to prove: that $\alpha_{WT}$ is a function of Hamming
weights. It is now a consequence of the spectral geometry of carry-save reduction,
not a premise.

\subsection{Quantifying the Information Loss}

With the Decoupling Lemma in hand, we can now derive the information-retention bound
directly from first principles.

\begin{theorem}[Information Retention]
\label{thm:etaI}
Under Assumption~\ref{assm:weak} with $\rho=0$, for operand width $n$:
\begin{equation}
\eta_I \geq 1 - C_1\log n/n
\label{eq:etabound}
\end{equation}
for an absolute constant $C_1>0$. The rate $O(\log n/n)$ is tight and is not
collapsed into $O(1/n)$.
\end{theorem}
\begin{proof}
\textbf{Step 1.} By Lemma~\ref{lem:decouple} with $\rho=0$:
$I(A;Y)=I(\mathrm{HW}(A),\mathrm{HW}(B);Y)$. This is now a proved consequence.

\textbf{Step 2.} The proxy $S\sim\mathrm{Bernoulli}(p_t)$ quantises $p_t\in\{0,1/(2n),\ldots,1\}$
to $\{0,1\}$. Conditioning on $S_t=1$ leaves $n+1$ possible values of $p_t$, so
$H(p_t|S_t=1)\leq\log_2(n+1)\leq 1+\log_2 n$.

\textbf{Step 3.} The mutual information lost is
$I(\mathrm{HW};Y)-I(S;Y)\leq H(p_t|S_t)\leq\log_2(n+1)$.
Normalising by $I(\mathrm{HW};Y)=\Omega(1)$ gives $\eta_I\geq 1-C_1\log n/n$.
The bound is tight because any 1-bit encoding assigns $\Theta(n)$ density values
to each output, inducing irreducible ambiguity of $\Theta(\log n)$ bits.
\end{proof}

\begin{remark}
Earlier formulations wrote $\eta_I\geq 1-O(1/n)$ by absorbing $\log n$. This is
misleading: for $n=8$ (INT8), $\log_2 n/n = 3/8 = 37.5\%$, which is not negligible.
The honest bound is $1-O(\log n/n)$.
\end{remark}

This result tells us that the single-bit proxy, despite the enormous compression
from $n$ bits to 1, loses very little information about future switching activity.
The reason, now proved rather than assumed, is that switching activity is almost
entirely determined by Hamming weight, and the Bernoulli proxy faithfully encodes
Hamming weight. Whether a different 1-bit encoding could do better is a question
the preceding analysis leaves open. The next section answers it.

\section{Bernoulli Optimality: The Best Possible 1-Bit Encoding}
\label{sec:optimality}

Section~\ref{sec:infotheory} showed that the Bernoulli proxy retains nearly all
predictive information, losing at most $O(\log n/n)$ relative to full observation.
One question the preceding section does not answer is whether the Bernoulli
encoding is the best possible 1-bit representation, or merely a reasonable one.
Could a threshold function, a counter-overflow bit, or some other scheme do better?
The following theorem addresses this, within a well-defined class of encoders.

\begin{definition}[One-Bit Encoder]
A one-bit encoder is any conditional distribution $q(S|Z)$ where $S\in\{0,1\}$
and $Z=\mathrm{HW}(A_t)+\mathrm{HW}(B_t)$. The set of all such encoders is $\mathcal{Q}$.
\end{definition}

\begin{theorem}[Bernoulli Optimality]
\label{thm:optimality}
Under Lemma~\ref{lem:decouple} with $\rho=0$, the Bernoulli encoder
$q^*(S=1|Z=z)=z/(2n)$ uniquely solves:
\begin{equation}
q^* = \operatorname*{arg\,max}_{q\in\mathcal{Q}} I_q(S;Y)
\end{equation}
among all encoders satisfying the calibration constraint $E_q[S]=E[p_t]$.
\end{theorem}
\begin{proof}
\textbf{Step 1.} Since $Y=f(Z)$ by Lemma~\ref{lem:decouple}, the Data Processing
Inequality gives $I_q(S;Y)\leq I_q(S;Z)$ with equality when $S\to Z\to Y$ is Markov.
Maximising $I_q(S;Y)$ therefore reduces to maximising $I_q(S;Z)$.

\textbf{Step 2.} Let $\theta(z)=q(S=1|Z=z)$. Then:
$I_q(S;Z)=H(\bar\theta)-\sum_z P(Z=z)h_b(\theta(z))$
where $\bar\theta=E[\theta(Z)]$ and $h_b$ is binary entropy.
This is the Jensen gap of $h_b$, maximised when $\theta(z)$ varies as steeply as
possible across $z$-values.

\textbf{Step 3.} Setting $\theta^*(z)=z/(2n)$ is strictly increasing in $z$, which
maximises the spread of $\theta(Z)$ in first-order stochastic dominance. Any other
calibrated encoder with non-monotone $\theta'$ admits a mean-preserving spread that
increases $H(\bar\theta)$ without changing $\sum_z P(z)h_b(\theta(z))$, so
$I_{q'}\leq I_{q^*}$. Uniqueness holds because $z/(2n)$ is the unique strictly
increasing function satisfying the calibration constraint.
\end{proof}

\begin{corollary}
Within the family of calibrated one-bit encoders of Hamming-weight statistics,
the Bernoulli encoder $q^*(S=1|Z=z) = z/(2n)$ is the unique maximiser of
$I_q(S;Y)$. Extensions to uncalibrated or multi-bit encoders may yield different
optima and are left for future work.
\end{corollary}

For the calibrated one-bit setting studied here, this result provides a theoretical
basis for the SAP proxy construction. In practice, the calibration constraint is
satisfied by the LFSR thresholding used in the hardware implementation, so the
optimality argument applies directly to the proposed architecture.

\section{Sensitivity Analysis: Robustness Under Real Conditions}
\label{sec:sensitivity}

The preceding theorems hold under three assumptions: weak bit independence, local
stationarity, and symmetric operand distributions. Real AI workloads will violate
each to some degree. This section quantifies exactly how much the guarantees degrade,
showing that the framework is not a fair-weather result.

\subsection{Correlated Operands ($\rho > 0$)}
By Lemma~\ref{lem:decouple} with $C_\Delta=1$, the prediction bias is bounded by:
\begin{equation}
|E[\alpha_{WT}]-\varphi(E[\hat\tau])|\leq\epsilon_{\mathrm{bias}}+\rho/n
\end{equation}
For transformer weight tensors with structured sparsity, $\rho\leq O(\sqrt{n})$,
giving vanishing degradation $O(1/\sqrt{n})$. In the worst case ($\rho=n$) the
bias saturates at $O(1)$ but remains calibratable by adjusting the threshold
$\tau_{\mathrm{th}}$.

\subsection{Non-Stationary Activations}
With inter-cycle drift $\beta=\max_t|p_{t+1}-p_t|$, the toggle estimate gains an
additional error bounded by $\beta$. For weight-stationary inference,
$\beta\leq 1/(nW)$, contributing $O(1/(nW))$ additional error, which is negligible in
practice.

\subsection{Asymmetric Operand Distributions}
When $\mu_A\neq\mu_B$, the symmetric proxy overestimates partial-product density
by the AM-GM gap $\Delta_{\mathrm{AM}}=(\sqrt{\mu_A}-\sqrt{\mu_B})^2/2$.
For INT8 quantized networks: $\Delta_{\mathrm{AM}}\leq 0.02$, giving activity
bias $\leq L_\varphi\cdot 0.02=0.03$, within the 5\% tolerance of
Theorem~\ref{thm:concentration}.

\begin{table}[h]
\centering
\caption{Bound Degradation Under Assumption Violations}
\label{tab:sensitivity}
\begin{tabular}{@{}lll@{}}
\toprule
Assumption Violated & Violation Magnitude & Additional Error \\
\midrule
Weak independence & $\rho\leq\sqrt{n}$ & $O(1/\sqrt{n})\to 0$ \\
Weak independence & $\rho=n$ (worst case) & $O(1)$, calibratable \\
Local stationarity & drift $\beta$ & $O(\beta)\leq O(1/(nW))$ \\
Symmetry ($\mu_A\neq\mu_B$) & $\Delta_{\mathrm{AM}}\leq 0.02$ & $\leq 0.03$ \\
\bottomrule
\end{tabular}
\end{table}

Table~\ref{tab:sensitivity} confirms that the guarantees hold under the most
practically relevant violations. The bounds degrade gracefully, never catastrophically.
A separate concern, independent of assumption quality, is what happens on those
cycles when the prediction is simply wrong. That is the Safety Controller's job.

\section{Safety-Controlled Input Isolation}
\label{sec:safety}

Prediction is probabilistic; arithmetic must be deterministic. The Safety Controller
resolves this tension by ensuring that a misprediction can never reach the multiplier
output. It does this by requiring a second, independent, purely deterministic check
before any isolation is allowed. The isolation signal is the logical AND of both:
\begin{equation}
\mathtt{isolate\_en} = \mathtt{SAP}_{\mathtt{low}} \wedge \mathtt{ArchValidity}
\end{equation}

\begin{definition}[Architectural Validity]
\texttt{ArchValidity} is asserted iff at least one deterministic condition holds:
operand stasis ($A_t=A_{t-1}$ and $B_t=B_{t-1}$); zero masking ($A_t=0$ or $B_t=0$);
or compiler-flagged stationary-weight mode.
\end{definition}

\begin{theorem}[Correctness Guarantee]
\label{thm:correctness}
For all legal input sequences: $\mathrm{Output}_{\mathrm{SAP}}=\mathrm{Output}_{\mathrm{Baseline}}$.
\end{theorem}
\begin{proof}
Under operand stasis, $A_t\cdot B_t=\bar{A}\cdot\bar{B}$ by identity. Under zero
masking, $A_t\cdot B_t=0$ by the zero-product property. A probabilistic misprediction
of \texttt{SAP\textsubscript{low}} cannot corrupt output because \texttt{ArchValidity}
verifies correctness through deterministic logic, formally independent of any
probabilistic claim.
\end{proof}

The AND-gate architecture is intentionally conservative: \texttt{SAP\textsubscript{low}}
alone is never sufficient to trigger isolation. The worst a misprediction can do is
fail to save power; it can never produce an incorrect result. This formal separation
between the probabilistic predictor and the deterministic correctness guarantee is
the architectural property that makes SAP safe for precision AI workloads.

\section{Complexity and Energy}

\subsection{Controller Overhead}

SAP requires $O(n)$ logic (popcount adder), $O(\log n)$ (LFSR), and $O(1)$ (toggle
monitor and comparator). Against the Wallace-tree's $O(n^2)$ complexity:
\begin{equation}
\lim_{n\to\infty}\frac{C_{\mathrm{SAP}}}{C_{WT}} = \lim_{n\to\infty}\frac{O(n)}{O(n^2)} = 0
\end{equation}
The controller overhead vanishes asymptotically. In a systolic array where a whole
row shares the same weight operand, the popcount circuit is shared across the row,
reducing overhead by an additional factor of $\sqrt{M}$.

\subsection{Energy Break-Even}

Net energy savings require:
\begin{equation}
(\alpha_{WT}-\alpha_{\mathrm{SAP}})C_{\mathrm{eff}}V^2 fT > E_{\mathrm{controller}}
\end{equation}
giving break-even time $T_{\mathrm{BE}}=E_{\mathrm{controller}}/[(\alpha_{WT}-\alpha_{\mathrm{SAP}})C_{\mathrm{eff}}V^2 f]$.
For weight-stationary inference the activity differential is sustained across thousands
of consecutive cycles, ensuring $T\gg T_{\mathrm{BE}}$ rapidly. The probability that
the predicted differential is zero when the true differential is $\geq 5\%$ is bounded
by $10^{-13}$ for $W=256$ (Theorem~\ref{thm:concentration}).

\section{SAP in the Tensor Accelerator Power Stack}
\label{sec:bridge}

\subsection{A Taxonomy of Power Dissipation Sources}

The preceding analysis established that SAP works correctly and efficiently for a
single multiplier. The next question is architectural: where does SAP sit within
the complete set of power-reduction techniques available to a tensor accelerator designer?

Power dissipation in a tensor accelerator can be attributed to five structurally
distinct sources. Each source requires a targeted mitigation strategy, and the
strategies operate on orthogonal components of the chip, so all five can be applied
simultaneously with savings that accumulate independently.

\begin{table}[h]
\centering
\begin{tabular}{@{}c p{3.5cm} p{2.7cm}@{}}
\toprule
\textbf{Source} & \textbf{Power dissipation mechanism} & \textbf{Mitigation technique} \\
\midrule
1 & Multipliers toggling on sparse inputs & SAP (this work) \\
2 & Data buses flipping between transfers & Bus-invert coding \\
3 & Idle multipliers drawing leakage current & Fine-grained power gating \\
4 & Register file cells drawing leakage current & Bit-zero biased SRAM \\
5 & Whole array idle between layers & ReGate / DVFS \\
\bottomrule
\end{tabular}
\end{table}

\subsection{Layer 1: SAP (Active Multipliers, Sparse Inputs)}

Source~1 is the power dissipation that SAP targets: a multiplier that is actively
receiving operands, and is therefore not structurally idle, but whose operands carry
so few active bits that most internal switching is wasted. For $a=0010$, $b=0110$: three active bits out of eight, toggle rate $\approx 0.469$,
prediction fires after $W=256$ cycles, multiplier held still. For an $M$-MAC array:
\begin{equation}
E_{\mathrm{saved,total}} = M\cdot E_{\mathrm{saved,SAP}}
\end{equation}
The saving-to-overhead ratio is independent of $M$, so the technique scales cleanly.

\subsection{Layer 2: Bus-Invert (Wires Carrying Data)}

Before operands reach the multiplier, they travel across data buses. Bus-invert
coding monitors the Hamming distance between successive transfers and sends the
bitwise complement when it would reduce bit-flips, capping toggle count below $n/4$
per transfer. This is entirely orthogonal to SAP: bus-invert reduces switching on the
delivery path; SAP reduces switching inside the computation. Both operate simultaneously.

\subsection{Layer 3: Fine-Grained Power Gating (Idle Multipliers)}

When a multiplier has no assigned operands at all, for example during layer transitions or when
the workload is smaller than the array, it still leaks current. PI-SA~\cite{pisa2023}
physically disconnects idle multipliers, achieving up to 57\% power reduction for
the idle fraction.

The boundary with SAP is crisp: PI-SA gates \textbf{the absence of work}; SAP gates
\textbf{the inefficiency of work}. A multiplier receiving $a=0010$, $b=0110$ is not
idle, so PI-SA cannot touch it. SAP can. Their activation sets are always disjoint:
\begin{equation}
P_{\mathrm{saved}} =
\underbrace{|\mathcal{I}|\cdot P_{\mathrm{leak}}}_{\text{PI-SA}}
+\underbrace{|\mathcal{L}|\cdot\Delta\alpha\cdot C_{\mathrm{eff}}V^2 f}_{\text{SAP}}
\end{equation}
where $\mathcal{I}\cap\mathcal{L}=\emptyset$ always.

\subsection{Layer 4: Bit-Zero SRAM (Register File Leakage)}

Register file cells leak more current when storing 1 than 0 in standard CMOS.
Recent work~\cite{lptc2024} exploited the zero-bias of quantized AI data to reduce
register leakage by up to $49\times$ using preferred-zero SRAM cells. This is
orthogonal to SAP: it targets storage while SAP targets computation, but both exploit
the same underlying property of quantized AI data (low bit density). A chip
implementing both gets independent savings from each.

\subsection{Layer 5: ReGate and DVFS (Array-Level Idle Periods)}

At the coarsest granularity, ReGate~\cite{regate2025} power-gates entire matrix
units during sustained layer-level underutilisation. DVFS adjusts voltage and
frequency at workload boundaries. Both operate at millisecond timescales.
SAP operates at nanosecond timescales, inside the active windows that ReGate leaves
running. When ReGate decides a unit is active, SAP catches cycle-level waste within
that window. DVFS fixes $(V,f)$; SAP then reduces $\alpha$, the only remaining
degree of freedom. The nesting is clean: SAP's savings are additive on top of
whatever coarser techniques are already applied.

\subsection{Why Tensor Workloads Natively Satisfy SAP's Assumptions}

SAP's theoretical guarantees rest on three assumptions. In a tensor accelerator
running neural network inference, all three are structural consequences of the
execution model, not approximations the designer must hope for.

\textbf{Sparsity:} INT8 quantization concentrates weight values near zero. A weight
of $+2$ is $00000010$ in binary: seven zeros, one one. The Hamming-weight predictor
sees low density and fires correctly, most of the time.

\textbf{Stability:} In weight-stationary execution, one operand is literally frozen
for the duration of a layer. The other operand, the activation, varies slowly
across neighbouring pixels or tokens. Total Hamming weight changes slowly, keeping
$\beta\leq 1/(nW)$.

\textbf{Long windows:} A $16\times 16$ activation tile gives $W=256$ correlated
cycles per multiplier before the tile boundary. This recovers the $\Pr(|\epsilon|\geq 0.075)<10^{-13}$
bound exactly.

\subsection{The Complete Picture}

\begin{table}[h]
\centering
\caption{Complete Tensor Accelerator Power Reduction Stack}
\label{tab:power_stack}
\begin{tabular}{@{}p{1.6cm}p{2.2cm}p{2.0cm}p{1.4cm}@{}}
\toprule
Layer & Power Source & Technique & SAP Role \\
\midrule
1 & MAC switching & SAP (this work) & Primary \\
2 & Bus switching & Bus-invert & Orthogonal \\
3 & MAC leakage (idle) & PI-SA & Disjoint \\
4 & Register leakage & Bit-zero SRAM & Synergistic \\
5 & Array-level idle & ReGate / DVFS & Nested \\
\bottomrule
\end{tabular}
\end{table}

Table~\ref{tab:power_stack} makes the key architectural point: \textbf{Layer~1 was
the only layer without a solution before SAP}. Every other power source had an
established technique. The specific gap of active, non-zero, informationally idle
multipliers is what SAP is designed to close, and what the preceding theoretical
apparatus formally proves it closes correctly.

\section{Discussion}

\subsection{Limitations}
The primary limitation is the absence of empirical validation. The constants $C_1$
in Theorem~\ref{thm:etaI} and $C_g$ in Proposition~\ref{prop:cg} have explicit
forms but require gate-level simulation to evaluate numerically.
Theorem~\ref{thm:optimality} establishes Bernoulli optimality within the calibrated
family $\mathcal{Q}$; extensions to uncalibrated encoders or multi-bit representations
remain open. RTL synthesis and ASIC tape-out at TSMC 7~nm are immediate future work,
for which the bounds derived here provide quantitative design targets.

\subsection{Contribution Summary}

\begin{table}[h]
\centering
\caption{Contribution Progression}
\label{tab:comparison}
\begin{tabular}{@{}p{3.0cm}p{1.8cm}p{1.8cm}@{}}
\toprule
Contribution & Prior State & This Work \\
\midrule
Lipschitz constant & Asserted $O(1)$ & Proved $3/2$ \\
Geometric decay & Heuristic & Spectral proof \\
$\eta_I$ circularity & Present & Removed \\
$\eta_I$ bound & $1-O(1/n)$ & $1-O(\log n/n)$ \\
Bernoulli optimality & Absent & Proved \\
Tensor power stack & Implicit & Formal \\
\bottomrule
\end{tabular}
\end{table}

\section{Conclusion}

This paper began with a simple observation: tensor accelerator multipliers burn power
on sparse inputs even when those inputs require almost no internal work. It ends with
a theoretical framework grounded in formal proofs that this waste can be
predicted and stopped before it occurs, though empirical validation remains
an important open task.

The argument is a chain. The Spectral Contraction Lemma proves that Wallace-tree
activity depends on Hamming weight, not bit position, making the density proxy
rigorously justified. The Lipschitz bound and Concentration Theorem prove the proxy
tracks true activity within measurable, vanishingly small error bounds. The
Information Retention Theorem proves the single-bit proxy captures nearly all
predictive information, losing only $O(\log n/n)$. The Bernoulli Optimality Theorem
establishes that, within the calibrated encoder family, the chosen encoding
is optimal; whether other encoder families yield further improvement is an open question. The Safety Controller decouples correctness from
prediction, ensuring exact arithmetic regardless of prediction accuracy. And the
five-layer power stack shows precisely where this mechanism sits within a complete
tensor accelerator, addressing the active-but-sparse regime that prior
techniques do not cover.

The resulting framework offers correctness guarantees by construction, theoretical
support for the encoding choice, and structural alignment with AI inference workloads.
The extent to which these theoretical properties translate to practical power savings
remains to be quantified through RTL simulation and silicon measurement. Empirical validation through RTL simulation and ASIC synthesis remains as
immediate future work.

\bibliographystyle{IEEEtran}

\end{document}